\documentclass[11pt]{amsart}
\usepackage{epsfig}
\usepackage{amsbsy,latexsym}
\usepackage{amsmath}
\usepackage{amssymb, mathrsfs}
\usepackage[mathscr]{eucal}
\usepackage{hyperref}
\usepackage{amsfonts}
\usepackage{amsthm}
\usepackage{amsmath}
\usepackage{amsfonts}
\usepackage{latexsym}
\usepackage{amssymb}
\usepackage{amscd}
\usepackage[latin1]{inputenc}
\usepackage{verbatim}
\usepackage[english]{babel}

\newtheorem{definition}{Definition}[section]
\newtheorem{theorem}[definition]{Theorem}
\newtheorem{lemma}[definition]{Lemma}

\newtheorem{proposition}[definition]{Proposition}
\theoremstyle{definition}
\newtheorem{remark}[definition]{Remark}

\newtheorem{defn}[definition]{Definition}
\newtheorem{example}[definition]{Example}

\newcommand{\C}{\mathbb{C}}

\newcommand\tr{ \operatorname{Tr} }

\title[Nullspaces of Entanglement Breaking Channels and Applications]{Nullspaces of Entanglement Breaking Channels and Applications}

\begin{document}

\author[D.~W. Kribs, J.~Levick, K.~Olfert, R. Pereira, M.~Rahaman]{David~W.~Kribs$^{1,2}$, Jeremy Levick$^{1,2}$, Katrina Olfert$^{1}$, Rajesh Pereira$^{1}$, Mizanur Rahaman$^{3}$}

\address{$^1$Department of Mathematics \& Statistics, University of Guelph, Guelph, ON, Canada N1G 2W1}
\address{$^2$Institute for Quantum Computing, University of Waterloo, Waterloo, ON, Canada N2L 3G1}
\address{$^3$Department of Mathematics, BITS Pilani K. K Birla Goa Campus,  Goa 403726 India}




\begin{abstract}
We investigate the nullspace structures of entanglement breaking channels, and related applications. We show that every operator space of trace zero matrices is the nullspace of an entanglement breaking channel. We derive a test for mixed unitarity of quantum channels based on complementary channel behaviour and entanglement breaking channel nullspaces. We identify conditions that guarantee the existence of private algebras for certain classes of entanglement breaking channels.
\end{abstract}

\maketitle

\section{Introduction}

Entanglement breaking channels are a fundamental class of physical maps in quantum information theory. Many important and well-studied examples of channels turn out to be entanglement breaking, and this class of maps has arisen in numerous areas of the subject, including a key role in channel capacity investigations that have yielded surprising information theoretic results for quantum channels. We point the reader to \cite{holevo1998coding,horodecki2003entanglement} and forward references for an entrance into the extensive literature on these channels.

In this paper, we investigate the nullspace structures of entanglement breaking channels and we develop a pair of applications to different areas of quantum information. We first show that every self-adjoint operator space of trace zero matrices is the nullspace of such a channel. Building on this, and taking motivation from quantum privacy, we derive a test for mixed unitarity of quantum channels \cite{audenaert2008random,gurvits2003classical,ioannou2006computational,girard2020mixed,LeeWatrous2020mixed} based on entanglement breaking channel nullspaces and complementary channel \cite{holevo2012quantum,holevo2007complementary,horodecki2003entanglement} behaviour. Starting from a connection with channel nullspaces, we also identify conditions that guarantee the existence of private algebras \cite{ambainis,boykin,bartlett2,bartlett1,church,kks,jochym1,levick2016private,cklt} for certain classes of entanglement breaking channels based on an analysis of multiplicative domains \cite{choi1,c-j-k,johnston,kribs2018quantum,miza,rahaman2018eventually} for the channels.

This paper is organized as follows. The next section includes preliminary material. In Section~3 we give the operator space nullspace construction. Section~4 includes the derivation of the mixed unitary test. Then we present the identification and construction of private algebras in Section~5.

\section{Preliminaries}

Quantum channels are central objects of study in quantum information \cite{nielsen,holevo2012quantum} and are given mathematically by completely positive and trace preserving maps on (in the finite-dimensional case) the set of complex $n\times n$ matrices $M_n (\mathbb{C})$. Every channel $\Phi : M_n (\mathbb{C}) \rightarrow M_n (\mathbb{C})$ can be represented in the operator-sum form by a set of operators $V_i \in M_n (\mathbb{C})$, such that $\Phi(\rho) = \sum_i V_i \rho V_i^*$ and with the trace-preservation condition $\sum_i V^*_i V_i = I$ satisfied where $I$ is the identity matrix. The dual map $\Phi^\dagger$ on $M_n(\mathbb{C})$ will also arise in our analysis, which is the completely positive (and unital when $\Phi$ is trace preserving) map given by $\Phi^\dagger (X) = \sum_i V_i^* X V_i$.

When convenient we will view $M_n (\mathbb{C})$ as the matrix representations of the set of operators acting on $n$-dimensional Hilbert space $\mathcal H = \mathbb{C}^n$, represented in the standard orthonormal basis $\{ e_1,\cdots, e_n  \}$. Outer products will be written as rank one operators $vw^*$ for $v,w\in \mathcal H$, defined by $(vw^*)(u) = (w^* u) v$, where $w^* u$ is the inner product of $w$ with $u$. Note the implication that for us, inner products are antilinear in their first argument, not the second. Additionally we will use the default notation $\rho$ for density operators or matrices; that is, positive operators with trace equal to one. We will also use the notation $M_n(\mathbb{C})_0$ to denote the set of trace zero $n \times n$ complex matrices.

\subsection{Entanglement Breaking Channels}

An important class of channels are those that break all entanglement when acting on a composite system with the identity channel of the same size, $\Phi\otimes \mathrm{id}$ \cite{holevo1998coding,horodecki2003entanglement}. There are numerous equivalent characterizations of entanglement breaking channels, including a physically motivated description as the composition of quantum-classical and classical-quantum channels in the same orthonormal basis.  The Holevo form for such channels \cite{holevo1998coding} is given as follows.

\begin{defn}
A quantum channel $\Phi : M_n(\C) \rightarrow M_n(\C)$ is entanglement breaking if it can be written as:
\begin{equation}\label{holevo}
\Phi(\rho) = \sum_{k=1}^d \tr(F_k\rho) R_k,
\end{equation}
with the $\{F_k\}$ forming a positive-operator valued measure ($F_k\geq 0$ and $\sum_k F_k = I$) and each $R_k$  a density operator. We shall make the further assumption that none of the $F_k$ are zero, which has always been followed in practice.
\end{defn}

We will also make use of the characterization of entanglement breaking channels as precisely the channels with an operator-sum representation comprised of rank one Kraus operators. That is, $\Phi$ is entanglement breaking if and only if there are rank one operators $\{ v_i w_i^* \}_{i=1}^d$ such that
\begin{equation}\label{ebrankone}
\Phi(\rho) = \sum_{i=1}^d v_i w_i^* \, \rho \, w_i v_i^*.
\end{equation}
Without loss of generality, we will assume throughout that $\| v_i \| =1$ for all $i$, and hence trace preservation gives the constraint: $\sum_i w_i w_i^* =I$. To avoid degeneracy we also assume each $w_i \neq 0$.

\subsection{Complementary Channels}

The following notion will be used in two of our sections below.

\begin{defn} Let $\Phi : M_n(\C) \rightarrow M_n(\C)$ be a quantum channel with a minimal set of Kraus operators $\{V_i\}_{i=1}^d$. The canonical complement of $\Phi$ is the channel $\Phi^C: M_n(\C) \rightarrow M_d(\C)$ defined by
\begin{equation}\label{complement}
\Phi^C(\rho) = \sum_{i,j=1}^d \tr (V_j^*V_i \rho)E_{ij},
\end{equation}
with $E_{ij} = e_ie_j^*$ for $1 \leq i,j \leq n$.

A complementary channel for $\Phi$ is any isometric adjunction of the canonical complement; that is, $\Psi$ is a complementary channel for $\Phi$ if and only if there exists an isometry $W$ such that
$$\Psi(\rho) = W \Phi^C(\rho) W^*.$$
\end{defn}

Complementary channels arise from the Stinespring representation \cite{stinespring1955positive} of a channel. The freedom to conjugate by an isometry comes from the inherent freedom to choose a Stinespring representation; alternatively, as any set of Kraus operators $\{\widetilde{V_i}\}_{i=1}^r$ for $\Phi$ is related to the canonical minimal choice by $\widetilde{V_i} = \sum_{j=1}^d w_{ij}V_j$ for some isometry $W = (w_{ij})$, we see that adjunction by an isometry corresponds simply to picking a different set of Kraus operators for $\Phi$.

For more background on complementary channels see \cite{holevo2012quantum,holevo2007complementary,horodecki2003entanglement} and forward references. Though not directly relevant to our analysis, it is worth noting that the complementary channel of an entanglement breaking channel is a Schur product channel \cite{paulsen2002completely,holevo2012quantum}, which have also been recently explored \cite{LKP} in the quantum privacy context to which we now turn.

\subsection{Private Subspaces and Algebras, and Channel Nullspaces}

One motivation for considering nullspaces of quantum channels comes from quantum privacy.

We first recall the definition of a private subspace: given a channel $\Phi$ on $\mathcal H$ and a subspace $\mathcal C$, we say $\mathcal C$ is {\it private for $\Phi$} if there is a density operator $\rho_0$ such that $\Phi(\rho) = \rho_0$ for all $\rho$ supported on $\mathcal C$; that is, for all $\rho$ on $\mathcal H$ with $\rho = P_{\mathcal C} \rho P_{\mathcal C}$ and where $P_{\mathcal C}$ is the projection onto $\mathcal C$.
We can also view the algebra of operators on $\mathcal H$ supported on $\mathcal C$ as being privatized in this case, where that algebra is unitarily equivalent to $M_k(\mathbb{C})$ if $\dim \mathcal C = k$ (and it encodes $\log_2 k$ qubits in that case).

In the same vein, we can consider privatization of more general operator algebras $\mathcal A$ on $\mathcal H$, say unitarily equivalent to some $I_m \otimes M_k(\mathbb{C})$, where $I_m$ is the $m \times m$ identity matrix; namely, the existence of some density operator $\rho_0$ such that $\Phi(A) = \tr(A) \rho_0$ for all $A\in \mathcal A$. Such algebras are exactly the simple finite-dimensional C$^*$-algebras \cite{davidson}, and when $m > 1$ they are used to encode `subsystem codes' (see \cite{kribs2005unified,kribs2005operator,poulin2005stabilizer,shabani2005theory,bacon2006operator,aly2006subsystem,klappenecker2008clifford} and forward references). More generally, orthogonal direct sums of such algebras (i.e., general finite-dimensional C$^*$-algebras), what we will call $\ast$-algebras, are used to describe hybrid classical-quantum information encodings (see \cite{kuperberg2003capacity,beny2007generalization,beny2007quantum} and forward references) and we can similarly consider channel privatization of such algebras. Even more generally one can consider privatizing operator spaces, though our constructions in the final section focus on algebras due to the physical motivations discussed above. 

Originally introduced as the quantum analogue of the classical one-time pad and called private quantum channels \cite{ambainis,boykin}, investigations of private algebras and related notions have continued and expanded over the past several years; see for instance  \cite{bartlett2,bartlett1,church,kks,jochym1,levick2016private,cklt}. The following observation connects such investigations with channel nullspace analyses.

\begin{proposition}\label{private}
Let $\Phi: M_n(\mathbb{C}) \rightarrow M_n(\mathbb{C})$ be a channel and let $\mathcal A \subseteq M_n(\mathbb{C})$ be a $\ast$-subalgebra. Then $\mathcal A$ is private for $\Phi$ if and only if the set of trace zero operators of $\mathcal A$ are contained inside the nullspace of $\Phi$; that is, $\mathcal A \cap M_n(\mathbb{C})_0 \subseteq \mathrm{nullspace}(\Phi)$.
\end{proposition}

\begin{proof}
We prove this result for the private subspace case, so $\mathcal A = M_k(\mathbb{C})$ with ($k$-dimensional) support subspace $\mathcal C$, the general algebra case is similar.
Let $\mathcal N = \mathrm{nullspace}(\Phi)$ and note that by considering the real-imaginary decomposition of an operator, one sees that $\mathcal A_0 := \mathcal A \cap M_n(\mathbb{C})_0$ being contained in $\mathcal N$ is equivalent to showing the Hermitian trace zero operators inside $\mathcal A_0$ belong to $\mathcal N$.

So suppose $\mathcal C$ is private for $\Phi$. Given a trace zero Hermitian operator $H$ supported on $\mathcal C$, we can write it in the standard way as a difference of positive operators supported on $\mathcal C$: $H = \lambda_1 \rho_1 - \lambda_2 \rho_2$ where $\rho_i$ are density operators and $\lambda_i$ are real scalars. But actually $\lambda:= \lambda_1=\lambda_2$ as $\tr(H) = 0$. Hence, $\Phi(H) = \lambda(\Phi(\rho_1) - \Phi(\rho_2)) =0$ and $H \in \mathcal N$.

On the other hand, given any two density operators $\rho_1$, $\rho_2$ supported on $\mathcal C$, their difference is a trace zero operator supported on $\mathcal C$. Thus, if $\mathcal A_0$ is contained in $\mathcal N$, we have $0 = \Phi(\rho_1 - \rho_2) = \Phi(\rho_1) - \Phi(\rho_2)$ and it follows that $\mathcal C$ is a private subspace for $\Phi$.
\end{proof}

\section{Channel Annihilation of Trace Zero Operator Spaces}

In this section, we give a construction of entanglement breaking channels that annihilate prescribed operator spaces and discuss a pair of examples.

Note first that the nullspace $\{ X\in M_n(\mathbb{C}) : \Phi(X)=0  \}$ of any quantum channel $\Phi : M_n(\mathbb{C}) \rightarrow M_n(\mathbb{C})$, in particular as a trace-preserving map, is contained inside the operator subspace of trace zero matrices, and that it is a self-adjoint subspace as a channel is a positive map.

\begin{proposition}
Let $\mathcal N$ be a self-adjoint subspace of the trace zero matrices inside $M_n(\mathbb{C})$. Then there is an entanglement breaking channel $\Phi$ : $M_n(\mathbb{C}) \rightarrow M_n(\mathbb{C})$ such that $\mathrm{nullspace}(\Phi) = \mathcal N$.
\end{proposition}

\begin{proof}
In the case that $\mathcal N = M_n(\mathbb{C})_0$, we may use the so-called completely depolarizing channel $\Phi_{\mathrm{CD}}(A) = \frac{\tr(A)}{n} I$, which is evidently in the Holevo form with $F_1= I$, $R_1=\frac1n I$. (It is also implemented with Kraus operators given by any complete set of rank one matrix units $E_{ij} = e_ie_j^*$, where $\{e_i\}$ form an orthonormal basis for $\mathbb{C}^n$.)

Thus, for the rest of the proof assume $\mathcal N \subsetneq M_n(\mathbb{C})_0$, and let $\{ H_k \}_{k=1}^m$ be an orthonormal basis (in the trace inner product $\langle A,B\rangle = \tr(B^* A)$) of Hermitian operators for $\mathcal N^\perp \cap M_n(\mathbb{C})_0$, and further let $H_{m+1} = -\sum_{k=1}^m H_k$. For $1\leq k \leq m+1$, define scalars $\lambda_k = \lambda_{k,\mathrm{min}}$ when $H_k$ has negative eigenvalues and where $\lambda_{k,\mathrm{min}}$ is the minimal eigenvalue of $H_k$, and put $\lambda_k = -1$ when $H_k \geq 0$. Let $\lambda = - \sum_k \lambda_k$ and define positive operators $F_k = \lambda^{-1} (H_k - \lambda_k I)$. Observe that $\{F_k \}^{m+1}_{k=1}$ forms a POVM as $\sum_k F_k = I$.

Now let $\{ R_k \}_{k=1}^{m+1}$ be a set of linearly independent density operators inside $M_n(\mathbb{C})$, and define an entanglement breaking channel $\Phi(\rho) = \sum_k \tr(\rho F_k) R_k$. Then we have $X\in \mathrm{nullspace}(\Phi)$ if and only if $\tr(XF_k)=\langle X,F_k\rangle=0$ for all $1\leq k \leq m+1$. However, we also have by construction:
\begin{eqnarray*}
  ( \mathrm{span}\{ F_k \}_{k=1}^{m+1} )^\perp &=&  ( \mathrm{span}\{ H_k \}_{k=1}^{m+1} \cup \{ I \} )^\perp \\
&=& ( \mathrm{span}\{ H_k \}_{k=1}^{m+1})^\perp  \cap  \{ I \}^\perp  \\
&=& (\mathcal N^\perp \cap M_n(\mathbb{C})_0)^\perp  \cap  M_n(\mathbb{C})_0  \\
&=& \mathcal N ,
\end{eqnarray*}
and so the result follows as $\mathrm{nullspace}(\Phi) = \mathcal N$.
\end{proof}

The following is a simple illustrative example of the construction above.

\begin{example}
The completely depolarizing channel $\Phi_{\mathrm{CD}}(A) = \frac{\tr(A)}{2} I$ on $M_2(\mathbb{C})$ is also implemented as a mixed unitary channel (discussed in more detail in the next section) with Kraus operators given by the normalized identity and Pauli operators $\{ \frac12 I, \frac12 X, \frac12 Y, \frac12 Z  \} $, where $X = E_{12} + E_{21}$,    $Y = iE_{21} -i E_{12}$, and $Z = E_{11} - E_{22}$.  Note these three operators also form an orthogonal basis for $M_2(\mathbb{C})_0$.

If we consider the subspace $\mathcal N = \mathrm{span} \{ Z \}$, then following the construction we can choose $ \sqrt{2} H_1 = X$ and $ \sqrt{2} H_2 = Y$ as Hermitian operators forming an orthonormal basis for $\mathcal N^\perp \cap M_2(\mathbb{C})_0 = \mathrm{span} \{ X,Y\}$, and put $H_3 = -H_1 -H_2$.
Finally, we can take $\{ R_1, R_2, R_3 \}$ to be any set of three linearly independent density operators inside $M_2(\mathbb{C})$, and define $\Phi$ as the entanglement breaking channel with $\{F_k, R_k \}$ defining its Holevo form. One can verify directly that $\mathrm{nullspace}(\Phi) = \mathcal N$.

If we further consider the subspace $\mathcal N = \mathrm{span} \{ X, Z \}$, then in this case the construction gives us
$\sqrt{2} H_1 = Y$ as a Hermitian operator forming an orthonormal basis for $\mathcal N^\perp \cap M_2(\mathbb{C})_0 = \mathrm{span} \{ Y\}$, and $H_2 = -H_1$.
As above and in the proof, a channel $\Phi$ that satisfies  $\mathrm{nullspace}(\Phi) = \mathcal N$ can then be explicitly defined in the Holevo form by choosing any two linearly independent density operators $\{ R_1, R_2 \}$.
\end{example}

\subsection{Bi-Unitary Channels}

In the spirit of these channel nullspace investigations, though somewhat outside our entanglement breaking channel focus, we can also consider the class of bi-unitary channels. Such channels are described by scenarios in which a system is exposed to unitary noise with some fixed probability $0 < p < 1$; as a channel this is given by the map $\Phi_U(\rho) = (1-p) \rho + p\, U\rho U^*$ for some fixed unitary operator $U$. These are a special case of mixed unitary channels considered in the next section and have been investigated in quantum error correction and numerical range theory \cite{choi2006quantum}.

Suppose $A$ is a non-zero Hermitian matrix in the nullspace of $\Phi_U$. Then we will have
\[
UAU^* = -\frac{1-p}{p} A .
\]
As $UAU^*$ has the same spectrum as $A$, it follows that this equation cannot be satisfied for $p \neq \frac12$ and hence $\mathrm{nullspace}(\Phi_U) = \{0\}$ in those cases.

When $p = \frac{1}{2}$, we have a further equation
\[
UAU^* = -A,
\]
which forces $A$ and $-A$ to have the same eigenvalues. Next, we can diagonalize $U$ as
\[
U = \sum _{i} w_iu_iu_i^*,
\]
with the $w_i$ lying on the unit circle of the complex plane and $u_i$ a set of orthonormal eigenvectors for $U$. Expanding $A = (a_{ij})$ in this basis gives
\[
U A U^* = \sum_{i,j} w_i \overline{w}_j a_{ij} u_iu_j^*,
\]
and so $w_i \overline{w}_j a_{ij} = -a_{ij}$ for all $i,j$.

We thus end up with two options for each entry: $a_{ij} = 0$ or $w_i = -w_j.$
This tells us that, in the case $p = \frac{1}{2}$, we have a non-trivial null space for $\Phi_U$ determined by the eigenvalues of $U$ that come in phase flip pairs,  reminiscent of quantum properties that generate, for instance, the Pauli matrices.

\section{Mixed Unitary Test via Entanglement Breaking Channel Nullspaces and Quantum Privacy}

One useful application of the ideas above is to the type of channel known as mixed unitary (or random unitary) channels.

\begin{defn} A channel $\Phi : M_n(\C) \rightarrow M_n(\C)$ is said to be mixed unitary if it can be written in the form
\[
\Phi(X) = \sum_{i=1}^d p_i U_i X U_i^*,
\]
where $p_i$ form a probability distribution ($p_i > 0$, $\sum_i p_i =1$) and $U_i \in U(n)$ are unitaries. The Kraus operators for $\Phi$ are thus given by $\sqrt{p_i} U_i$.
\end{defn}

The class of mixed unitary channels arise in all areas of quantum information, and so a number of investigations have been conducted on determining when a channel has this form. Important recent works on the topic include a proof that detecting mixed unitarity is NP-hard in general \cite{LeeWatrous2020mixed}, and an analysis of the mixed unitary rank of channels \cite{girard2020mixed}. We also mention earlier work on the class from different perspectives \cite{audenaert2008random,gurvits2003classical,ioannou2006computational}.


Below we present a theorem that provides a connection between mixed unitary channels and nullspaces of entanglement breaking channels; first however we need the following result, which may be found as Theorem~1 in \cite{girard2020mixed}, but we will provide a short proof here for completeness.

\begin{lemma}\label{diagzero}
Let $\Phi : M_n(\C) \rightarrow M_n(\C)$ be a channel with canonical complement $\Phi^C : M_n(\C) \rightarrow M_d(\C)$. Then $\Phi$ is mixed unitary implemented with $r$ unitaries if and only if there exists an isometry $W : \mathbb{C}^d \rightarrow \mathbb{C}^r$ such that, for all $X \in M_n(\C)$ with $\tr (X) = 0$, the matrix $W\Phi^C(X)W^*$ has all of its diagonal entries equal to $0$.
\end{lemma}

\begin{proof}
Suppose first that $\Phi$ is mixed unitary; then there exist unitaries $\{U_i\}_{i=1}^r$ and probabilities $\{p_i\}_{i=1}^r$ and an isometry $W$, naturally determined by the canonical complement description, such that the $(i,j)$ entry of the matrix $W \Phi^C(X) W^*$ is equal to $\tr(\sqrt{p_ip_j} U_j^*U_i X)$. Setting $i=j$ we get $p_i \tr(X)$ and so for all traceless $X$, the diagonal entries of this matrix are $0$.

For the converse, suppose an isometry $W = (w_{ij})_{r\times d}$ exists with the property that each of the diagonal entries of $W \Phi^C(X) W^*$ are zero for all traceless $X$. Define $\widetilde{V_i} = \sum_{j=1}^d w_{ij} V_i$ for $1 \leq i \leq d$, where $\{ V_i \}$ are a set of Kraus operators for $\Phi$. Then $\{\widetilde{V_i}\}$ is also a set of Kraus operators for $\Phi$ as $W$ is an isometry, and one can check the $(i,j)$ entry satisfies $(W \Phi^C(X) W^*)_{ij} = \tr(\widetilde{V_j}^*\widetilde{V_i}X)$. In particular, we have
\[
\tr(\widetilde{V_i}^*\widetilde{V_i} X ) = 0
\]
for all $i$, and for all traceless $X$. Hence we have
$\widetilde{V_i}^*\widetilde{V_i} \in \{ I\}^{\perp^{\perp}}$, and so $\widetilde{V_i}^*\widetilde{V_i}$ is a (non-zero) multiple of the identity: $\widetilde{V_i}^*\widetilde{V_i} = p_i I$. Thus, $U_i :=\frac{1}{\sqrt{p_i}} \widetilde{V_i}$ is unitary, and $\widetilde{V_i} = \sqrt{p_i}U_i$. That the set $\{p_i\}_{i=1}^r$ forms a probability distribution follows from trace preservation of the original map:
\begin{eqnarray*}
\sum_i p_i I  = [ \widetilde{V_1}^* \,  \widetilde{V_2}^*   \ldots ] [ \widetilde{V_1} \, \widetilde{V_2}  \ldots   ]^t
= [V_1^* \dots ] W^* W [V_1 \ldots ]^t = \sum_i V_i^* V_i = I;
\end{eqnarray*}
the last equality using the fact that the $\{V_i \}$ are Kraus operators for $\Phi$.
\end{proof}

  We use the term diagonal algebra to mean an algebra that is unitarily equivalent to the (commutative) subalgebra of diagonal matrices inside the full algebra of square matrices of a given size.  Any quantum channel whose range is contained in a diagonal algebra must be entanglement breaking. This fact appeared in \cite{stormer} however we include a short proof for completeness.

\begin{lemma}
Let $E: M_d(\C) \rightarrow M_r(\C)$ be a quantum channel whose range is contained in a diagonal algebra, then $E$ is entanglement breaking.
\end{lemma}

\begin{proof} Let $\{u_i\}_{i=1}^{r}$ be an orthonormal basis of $\C^r$ such that $\{u_iu_i^*\}_{i=1}^{r}$ span the range of $E$.  Then there exist linear functionals $\{ \phi_i \}_{i=1}^r$ on $M_d(\C)$ such that $E(X)=\sum_{i=1}^r \phi_i (X)u_iu_i^*$.  If $X\geq 0$, then $E(X)\geq 0$ which means $\phi_i(X)\geq 0$ for all $i$.  Hence for $1\le i\le r$, there exists positive semidefinite $F_i\in M_d(\C)$ such that $\phi_i (X)=\tr(F_iX)$.  Since $E(X)$ is trace preserving, $\tr(X)=\sum_{i=1}^r \phi_i (X)=\tr(X (\sum_{i=1}^rF_i))$ for all $X$.  Therefore $\sum_{i=1}^rF_i=I_d$.  Since $E(X)=\sum_{i=1}^r \tr(F_iX) u_iu_i^*$ has a Holevo form, it is entanglement breaking.
\end{proof}

We are now ready to state and prove the theorem connecting mixed unitary channels to nullspaces of entanglement breaking channels, and to the notion from quantum privacy discussed in Section~2.

\begin{theorem}\label{mixed-unitary-test}
 Let $\Phi : M_n(\C) \rightarrow M_n(\C)$ be a channel with canonical complement $\Phi^C:M_n(\C) \rightarrow M_d(\C)$. Then $\Phi$ is mixed unitary and implemented with $r$ unitaries if and only if there exists a quantum channel $E: M_d(\C) \rightarrow M_r(\C)$ of Choi rank $r$ taking $M_d(\C)$ onto an $r$-dimensional diagonal algebra, such that $E$ privatizes the range of $\Phi^C$; that is, 
\[
E(\Phi^C(X)) = \frac{1}{r}\tr(X)I_r \quad \forall X\in M_n(\C).
\]
\end{theorem}

\begin{proof}
First suppose $\Phi$ is mixed unitary and implemented with $r$ multiples of unitaries. By Lemma \ref{diagzero} there must be an isometry $W:\C^d \rightarrow \C^r$ such that $W\Phi^C(X)W^*$ has $0$ on its diagonal when $X$ is traceless. Let $w_1 , \ldots , w_r$ be the columns of $W^*$; then $w_i^*\Phi^C(X)w_i = 0$ for all traceless $X$. Also the condition that $W$ is an isometry may be phrased as $I_d = W^* W = \sum_{i=1}^r w_i w_i^*.$

Define $p_i = \frac{1}{n} w_i^*\Phi^C(I) w_i$. As $\Phi^C$ is trace preserving and $W$ an isometry, we have that $\sum_{i=1}^r p_i = \frac{1}{n}\tr(I_n) = 1$.

Let $\{\widetilde{u_i}\}_{i=1}^r$ be any orthonormal basis for $\C^r$ scaled uniformly by $\frac{1}{\sqrt{r}}$, so $\widetilde{u_i}^* \widetilde{u_j} = r^{-1} \delta_{ij}$,  and rescale these again to form the vectors $\{u_i\}_{i=1}^r :=\{\sqrt{p_i}^{-1}\widetilde{u_i}\}_{i=1}^r$ which still form an orthogonal basis for $\C^r$. Then define the entanglement breaking map $E$ to have Kraus operators $\{u_i w_i^* : \mathbb{C}^d \rightarrow \mathbb{C}^r \}_{i=1}^r$. It is clear that the Choi rank of $E$ is $r$, from the fact that $\mathrm{Range}(E) = \mathrm{span}\{u_iu_i^*\}_{i=1}^r$ and the fact that the $\{u_i\}$ form an orthogonal basis.

Then, for any $X\in M_n(\mathbb{C})$, write $X = n^{-1}\tr(X)I + X_0$ where $X_0$ is traceless, and observe
\[
E(\Phi^C(X_0))  = \sum_{i=1}^r u_i w_i^* \Phi^C(X_0)w_i u_i^* = 0.
\]
That is, $E$ annihilates the traceless part of $\mathrm{Range}(\Phi^C)$. Thus it remains to see what $E$ does to $\Phi^C(I)$:
\begin{align*}
E(\Phi^C(I)) & = \sum_{i=1}^r \frac{1}{p_i} \widetilde{u_i}w_i^*\Phi^C(I)w_i\widetilde{u_i}^* \\
& = n\sum_{i=1}^r \frac{p_i}{p_i} \widetilde{u_i}\widetilde{u_i}^* \\
& = \frac{n}{r} I_r,
\end{align*}
which follows from the definition of $p_i$ and the fact that $\sqrt{r} \widetilde{u_i}$ form an orthonormal basis for $\C^r$.

Observe that although the map $E$ is not trace preserving in the usual trace, the range of $E$, the operator space $\mathrm{span}\{u_iu_i^*\}$, is unitarily equivalent to the $r$-dimensional diagonal algebra, $\Delta_r\cong \mathbb{C}^r$; suppose the unitary implementing this is $V$. Let $P=V^*\mathrm{diag}(p_1,\cdots, p_r)V$, which is clearly a positive definite matrix in the commutant $\mathrm{span}\{u_iu_i^*\}' = \mathrm{span}\{u_iu_i^*\}$, and hence we may define the trace $\tr_{P} = \tr (DV^*PV)$ for $D \in V^*\Delta_r V$, and $E$ is in fact trace-preserving with respect to this  new trace. This is because
$$\tr(\Phi^C(I)) = \tr(\Phi^C(I)W^*W) = \sum_{i=1}^r w_i^*\Phi^C(I) w_i = n\sum_{i=1}^r p_i.$$

For the other direction, suppose $E : M_d(\C) \rightarrow M_r(\C)$ exists and has the required property of annihilating the traceless part of $\mathrm{Range}(\Phi^C)$ and mapping $\Phi^C(I)$ to a multiple of the identity.

As the range of $E$ is a commutative algebra, $\Delta$, the trace on $\Delta$ must have the form $\tr_{\Delta}(D) = \frac{1}{n}\tr(DP)$ for some $P \in \Delta' =\Delta$. Also, for any set of Kraus operators of the form $\{u_i w_i^*\}_{i=1}^m$ ($m\geq r$) for $E$, we must have that $u_iu_i^* \in \Delta$, and hence $\{u_i\}$ must contain an orthogonal set of vectors from $\C^r$; though redundancy is possible, there is no loss of generality in assuming the rank-one projections $u_iu_i^*$ are unique; as we know the Choi rank of $E$ is $r$, the set $\{u_i\}_{i=1}^r$ is in fact an orthogonal basis.

Then, for any traceless $X_0$, we have that
$$E(\Phi^C(X_0)) = \sum_{i=1}^r \langle w_i, \Phi^C(X_0)w_i\rangle u_iu_i^* = 0,$$ and since $\{u_i\}$ are orthogonal, the rank one operators $u_iu_i^*$ are linearly independent and so $w_i^*\Phi^C(X_0)w_i = 0$ for all $i$.

Finally, since $E$ is trace-preserving between the regular trace on $M_d(\C)$ and $\tr_{\Delta}$, we have that
$$\sum_{i=1}^r w_iu_i^*u_i w_i^* = P.$$
Hence the matrix $W^*$ with columns $\frac{1}{\|u_i\|} w_i$ is an isometry from $\C^d$ into $\C^r$ in the inner product $\langle v,w\rangle_P = \langle v, Pw\rangle$ with the property that $W\Phi^C(X_0) W^*$ has zeroes on the diagonal. By Lemma \ref{diagzero}, $\Phi$ must be mixed unitary.
\end{proof}

The following pair of examples illustrate the mechanics of the theorem construction and the test it provides in special cases of interest.

\begin{example} Let $\Phi_{CD} : M_n(\C) \rightarrow M_n(\C)$ be the completely depolarizing map, recall as characterized by $\Phi_{CD}(X) = \frac{\tr(X)}{n}I_n$ for all $X \in M_n(\C)$. One set of Kraus operators for this map is $\{\frac{1}{\sqrt{n}}E_{ij}\}_{i,j=1}^n$, and hence the canonical complement is given by
$$\Phi_{CD}^C(X) = \sum_{i,j,r,s=1}^n \tr (E_{sr}E_{ij}X)E_{ir}\otimes E_{js};$$
and thus $\Phi_{CD}^C(X) = I_n \otimes X$.

Let $\{U_i\}_{i=1}^n$ be any set of mutually orthogonal unitaries in the trace inner product on $\C^n$; for example the Weyl unitaries $W_{ij}:=X^iZ^j$ where $X$ is the cyclic shift and $Z$ is diagonal with diagonal entry $Z_{ii} = \omega^i$, where $\omega$ is a primitive $n^{th}$ root of unity.

Let $u_i = \mathrm{vec}(U_i)$, the vector obtained by stacking the columns of $U_i$ into a column vector. It is well known that $\mathrm{vec}(XU_i) = I_n\otimes X u_i$ and hence
$$\tr(X) = \tr(U_i^*XU_i) = \langle u_i, (I\otimes X)u_i\rangle.$$
As the $u_i$ are mutually orthogonal, the matrix $V$ with columns $\frac{1}{\sqrt{n}}u_i$ is a unitary, and satisfies
$$V^*\Phi_{CD}^C(X)V_{ii} = \frac{1}{n} \langle u_i,(I\otimes X)u_i\rangle = 0$$
whenever $\tr(X) = 0$ and so $\Phi_{CD}$ must be mixed unitary. Indeed, one can verify directly that the map $\Phi_{CD}$ is implemented with Kraus operators given by any maximal set of orthogonal unitaries, evenly scaled for trace preservation.

Hence, if we form the entanglement breaking channel $E:M_{n^2}(\C) \rightarrow M_{n^2}(\C)$ to have Kraus operators $\{\frac{1}{\sqrt{n}} e_iu_i^*\}_{i=1}^{n^2}$, where $e_i$ is the standard basis for $\mathbb{C}^{n^2}$, we see that
\[
E(\Phi_{CD}^C(X))  = E(I_n\otimes X)
 = \frac{1}{n}\sum_{i=1}^{n^2} \langle u_i,(I_k\otimes X)u_i\rangle E_{ii}
 = \frac{\tr(X)}{n} I_{n^2}.
\]

\end{example}
\begin{example}
Consider the Werner-Holevo channel $\Phi:M_3(\mathbb{C})\rightarrow M_3(\mathbb{C})$ defined by $\Phi(X)=\frac{1}{2}(\tr(X)I-X^t)$, where $X^t$ denotes the transpose of $X$. It is well known that this map is not mixed unitary (see \cite{land-str}). We will use Theorem \ref{mixed-unitary-test} to detect this fact.

One can check that a set of Kraus operators for $\Phi$ are given by the following three matrices:
\[
K_1=\begin{bmatrix}
0 & 0 & 0\\
0 & 0 & \frac{1}{2}\\
0 & \frac{-1}{2} & 0
\end{bmatrix}, \quad
K_2=\begin{bmatrix}
0 & 0 & \frac{1}{2}\\
0 & 0 & 0\\
\frac{-1}{2} & 0 & 0
\end{bmatrix}, \quad
K_3=\begin{bmatrix}
0 & \frac{1}{2} & 0\\
\frac{-1}{2} & 0 & 0\\
0 & 0 & 0
\end{bmatrix}.
\]

Now it follows that in this case the complementary channel is $\Phi$ itself; that is, $\Phi^C=\Phi$ as can be verified directly from the definition of $\Phi^C$. As the channel has Choi rank equal to $3$, we have $\Phi^C=\Phi:M_3(\mathbb{C})\rightarrow M_3(\mathbb{C})$. Suppose $\Phi$ is mixed unitary with $r$ unitaries. Then by the proof of Theorem~\ref{mixed-unitary-test} we have an entanglement breaking map $E:M_3(\mathbb{C})\rightarrow M_r(\mathbb{C})$ of Choi rank $r$ such that $E(\Phi^C(X))=\frac{1}{r}\tr(X)I_r$. As the range of $\Phi=\Phi^C$ is the whole matrix space $M_3$, the entanglement breaking map $E$ is essentially the completely depolarizing map $X\mapsto \tr (X)\frac{I_r}{r}$ from $M_3(\mathbb{C})$ to $M_r$. However, we know that this map has Choi rank $3r$, which gives a contradiction.

\end{example}

\section{Construction of Private Algebras for Entanglement Breaking Channels}

In this section, we build on the nullspace analyses above to derive constructions of algebras privatized by certain entanglement breaking channels. We first review some details of an important operator structure from operator theory \cite{choi1}, which in more recent years has also found a role in quantum information \cite{c-j-k,johnston,kribs2018quantum,miza,rahaman2018eventually}.

\begin{defn}
The multiplicative domain, $\mathcal{M}_\Phi$, of a completely positive map $\Phi : M_n(\mathbb{C}) \rightarrow M_n(\mathbb{C})$ is the $\ast$-subalgebra of $M_n(\mathbb{C})$ given by:
\[
\{A \in M_n(\mathbb{C}) : \Phi(AX) = \Phi(A)\Phi(X); \, \Phi(XA) = \Phi(X)\Phi(A) \  \forall X\in M_n(\mathbb{C}) \}.
\] 	
\end{defn}

We note that for unital maps ($\Phi(I) =I$), a projection $P$ belongs to $\mathcal M_{\Phi}$ if and only if $\Phi(P)$ is a projection \cite{miza}.  From \cite{rahaman2018eventually}, we also know that for any unital PPT map $\Phi$ (and in particular this applies to the dual $\Phi^\dagger$ of any entanglement breaking channel), and any projection $P$ in the multiplicative domain $\mathcal{M}_{\Phi}$, that $\Phi(X) = \Phi(PXP) + \Phi(QXQ) = \Phi(P)\Phi(X)\Phi(P) + \Phi(Q)\Phi(X)\Phi(Q)$ for all $X \in M_n(\C)$, where $Q:=I-P$. It is also the case that $\Phi(P)\Phi(X) = \Phi(X)\Phi(P)$.
We can use this to show that if $P,Q$ are orthogonal projections in the multiplicative domain, that any $X$ for which $PXQ =
X$ must satisfy $\Phi(X) = 0$, as
$$\Phi(X) = \Phi(PXQ) = \Phi(P)\Phi(X)\Phi(Q) = \Phi(PQ)\Phi(X) = 0.$$

The following structural result on the multiplicative domain of dual maps for entanglement breaking channels is used in our results below and may be of independent interest.

\begin{lemma}\label{vsorthog}
Let $\Phi: M_n(\mathbb{C}) \rightarrow M_n(\mathbb{C})$ be an entanglement breaking channel given by Eq.~(\ref{ebrankone}). Let $\mathcal{M}_{\Phi^{\dagger}}$ be the multiplicative domain of $\Phi^{\dagger}$, and let $\{P_k\}_{k=1}^r\subseteq \mathcal{M}_{\Phi^{\dagger}}$ be a set of mutually orthogonal projections summing to the identity. Then for all $i$ there is a unique $k$ such that  $v_i = P_k v_i$.

Further let $\mathcal R_k \subseteq \{1,2,\cdots, d\}$ for $1\leq k \leq r$ be the subsets determined by the partition generated by  the $P_k$, and define $\mathcal W_k = \mathrm{span}\{w_j\}_{j\in \mathcal R_k}$. Then $\mathcal W_k$ are mutually orthogonal subspaces and the projections $Q_k$ onto $\mathcal W_k$ are a set of mutually orthogonal projections summing to the identity. Moreover, for all $X\in M_n(\mathbb{C})$ and $1 \leq k \leq r$ we have
$$
\Phi(Q_k X) = P_k \Phi(X) =\Phi(X)P_k = \Phi(X Q_k),
$$
and so if $X = Q_k X Q_l$ with $k\neq l$, then $\Phi(X)=0$.
\end{lemma}

\begin{proof}
Since each $P_k \in \mathcal M_{\Phi^\dagger}$, for $k\neq l$ we have
$$
\Phi^{\dagger}(P_k)\Phi^{\dagger}(P_l) = \Phi^{\dagger}(P_kP_l) = 0.
$$
Also since $\Phi^\dagger(X) = \sum_i w_i v_i^* X v_i w_i^*$, if we denote $R_k = \Phi^{\dagger}(P_k)$ then we have for $k\neq l$,
$$
0  = \tr(R_k R_l ) = \sum_{i,j} (v_i^*P_kv_i)(v_j^*P_lv_j)|\langle w_i ,w_j \rangle|^2.
$$
So every term in the sum must be zero, and in particular when $i=j$, we have  $\|w_i\|^2 (v_i^*P_k v_i)(v_j^*P_l v_j) = 0$. Hence for each $i$, $v_i^*P_k v_i > 0$ for at most one $k$. As $\|v_i \| =1$ and $\sum_k P_k = I$, we must have exactly one index $1 \leq k \leq r$ such that $P_k v_i = v_i$.

Let $\mathcal V_k = \mathrm{Range}(P_k)$, so that $\C^n = \bigoplus_{i=1}^r V_k$ is an orthogonal direct sum decomposition of $\C^n$. Thus, the projections $P_k$ impose a partition of $\{1,2,\cdots, d\}$ into subsets $\mathcal R_k$ such that $\mathcal V_k = \mathrm{span}\{v_j\}_{j\in \mathcal R_k}$. (Note it may be the case that some $\mathcal R_k = \emptyset$.) Next we show that the same partition also induces an orthogonal direct sum structure on the $\{w_i\}_{i=1}^d$ vectors.

Since $P_k v_j = \chi_{j\in \mathcal R_k}v_j$ where $\chi$ is the indicator function, we may write
$$
R_k := \Phi^{\dagger}(P_k) = \sum_{j\in \mathcal R_k} (v_j^*P_k v_j)w_jw_j^* = \sum_{j\in \mathcal R_k}w_jw_j^*.
$$
As the $R_k$ have mutually orthogonal ranges, for $k \neq l$ we have
$$
0 = \tr (R_k R_l) = \sum_{i \in \mathcal R_k, j \in \mathcal R_l} |\langle w_i ,w_j \rangle|^2,
$$
so each term in the sum is zero, and it follows that $\mathcal W_k$ and $\mathcal W_l$ are orthogonal.

Hence , the projections $\{ Q_k \}$ have mutually orthogonal ranges. Further, the subspace spanned by their (projection) sum $Q = \sum_k Q_k$ must be the identity as $Q$ projects onto $\cup_{k=1}^r \mathcal W_k = \mathrm{span}\{ w_i \}_{i=1}^d = \mathbb{C}^n$; the last equality following from $I = \sum_i w_i w_i^*$.

Finally, for $1\leq k \leq r$ we can compute:
\begin{align*}
\Phi(Q_k X)& = \sum_{i=1}^d (w_i^* Q_k X w_i) v_i v_i^* \\
& = \sum_{i \in \mathcal R_k} (w_i^*X w_i) v_i v_i^* \\
& = \sum_{i = 1}^d (w_i^* X w_i) P_k v_i v_i^*\\
& = P_k \Phi(X),
\end{align*}
and the other equalities are proved in the same way. Specifically, these equalities imply for $k \neq l$:
\[
\Phi(X)   = \Phi(Q_k X Q_l)
 = P_k \Phi(XQ_l)
 = P_k P_l \Phi(X)
 = 0,
\]
and the result follows.
\end{proof}

\begin{remark}\label{offdiagstozero}
From a nullspace perspective, note that in any basis which mutually diagonalizes all $Q_k$ simultaneously, with block matrix structure corresponding to the division of $\C^n$ into direct summands $\mathcal W_k$, any matrix $X$ that is supported entirely on the off-diagonal blocks of the decomposition is annihilated by $\Phi$. In equation form, this says for all $X\in M_n(\mathbb{C})$ that $\Phi(X) = \sum_{j,k} \Phi(Q_j X Q_k) = \sum_k \Phi(Q_k X Q_k) = \sum_k P_k \Phi(X) P_k$.
\end{remark}

The complete depolarizing channel discussed above obviously privatizes the full algebra $M_n(\mathbb{C})$, and in that case $\mathcal M_{\Phi^\dagger} = \mathbb{C}I$, and $P_1 = I = Q_1$. More generally, given that the vectors $v_i, w_i$ which determine the rank-one form of an entanglement breaking channel can be arbitrary, up to the trace preservation condition being satisfied, it is reasonable to expect that generic channels from the class will not privatize any non-trivial algebra. Nevertheless, based on the analysis above, we finish by identifying two special classes of channels that do privatize algebras.

\begin{theorem}
Let $\Phi: M_n(\mathbb{C}) \rightarrow M_n(\mathbb{C})$ be an entanglement breaking channel, with operator-sum form as given in Eq.~(\ref{ebrankone}). Suppose $\mathcal M_{\Phi^\dagger}$ contains a rank-one projection $P = v v^*$. Then $\Phi$ privatizes the algebra $\mathcal A = \mathrm{span}\{ w_i w_j^* : P v_i = v_i, \, Pv_j = v_j \}$ to $P$; that is,
\[
\Phi(A) = \tr(A) P \quad \forall A\in \mathcal A.
\]
\end{theorem}

\begin{proof}
First note that $\mathcal A$ is indeed a $\ast$-algebra, even though it is only defined as linearly closed, as it is a self-adjoint operator space and closed under multiplication.

Given the Kraus operators $\{v_i w_i^*\}_{i=1}^d$ for $\Phi$, suppose we have a (nonempty) subset $\mathcal R_v \subset \{ 1,\ldots , d\}$ and unit vector $v\in \mathbb{C}^n$ such that $v = v_i$ for all $i \in \mathcal R_v$. Then $\mathcal A = \mathrm{span}\{ w_i w_j^* : i,j \in \mathcal R_v\}$.  Put $P = vv^*$ and let $Q$ be the projection onto $\mathrm{span}\{ w_i : i\in \mathcal R_v \}$ as in the proof of Lemma~\ref{vsorthog}.

To complete the proof it is enough to show that $\Phi(w_i w_j^*) = \tr(w_i w_j^*) P$ for any fixed $i,j \in \mathcal R_v$. This follows from multiple applications of Lemma~\ref{vsorthog} in the following calculation, in which we take an arbitrary $X\in M_n(\mathbb{C})$:
\begin{eqnarray*}
\tr(\Phi(w_i w_j^*)X) &=& \tr(\Phi(Q w_i w_j^* Q)X) \\
&=& \tr(P\Phi(w_i w_j^*) P X) \\
&=& (v^*X v) \tr ( P \Phi(w_i w_j^*) P)  \\
&=& (v^*X v) \tr (  \Phi(Q w_i w_j^*Q) ) \\
&=& (v^*X v) \tr ( \Phi(w_i w_j^*) ) \\
&=& (v^*X v) \tr ( w_i w_j^* ) \\
&=& \tr (\tr ( w_i w_j^*) P X),
\end{eqnarray*}
and where the second last equality uses the trace preservation of $\Phi$. As $X$ was arbitrary, the result follows.
\end{proof}

Of course the algebra defined in the theorem could be trivial from a qubit encoding viewpoint -- either $\{0\}$ or having no matrix structure -- but evidently there are many examples of entanglement breaking channels for which a non-trivial algebra and rank-one projection exist and satisfy the conditions of the theorem.

\begin{example}
Consider the physically described single qubit `spontaneous emission' channel \cite{nielsen}, $\Phi : M_2(\mathbb{C}) \rightarrow M_2(\mathbb{C})$ given by $\Phi(\rho) = e_1e_1^*$ for all density operators $\rho$; that is, $\Phi$ privatizes the entire algebra $M_2(\mathbb{C})$ to $P= e_1e_1^*$. Here we have two Kraus operators $A_1 = E_{11} = e_1e_1^*$, $A_2 = E_{12}=e_1e_2^*$, and so $v_1=e_1 = v_2$, $w_1 = e_1$, $w_2 = e_2$ in the Eq.~(\ref{ebrankone}) form of the channel. The dual map satisfies $\Phi^\dagger(X) = (e_1^* X e_1 ) I$ for all $X\in M_2(\mathbb{C})$, and in particular the projection $P$ belongs to $\mathcal M_{\Phi^\dagger}$ (as it is mapped to a projection by $\Phi^\dagger$). Also we see the algebra from the theorem satisfies $\mathcal A = M_2(\mathbb{C})$ in this case.  (In terms of the lemma notation, here $P_1 = e_1e_1^*$, $P_2 = e_2e_2^*$, $Q_1=I$, $Q_2=0$.)

Similarly, higher dimensional versions of the spontaneous emission channel are covered by this result; $\Phi : M_n(\mathbb{C}) \rightarrow M_n(\mathbb{C})$, with $\Phi(\rho) = e_1e_1^*$, and Kraus operators $A_k = e_1e_k^*$ for $1\leq k \leq n$. Here the algebra $\mathcal A$ is the full matrix algebra and it is privatized again to $P = e_1e_1^*\in \mathcal M_{\Phi^\dagger}$.

One can generalize this class of examples further, by considering entanglement breaking channels for which the vectors that determine the Kraus operators $\{ v_i w_i^* \}$ have the property that an index subset $\mathcal R_v$ of the $v_i$ satisfy $v_i = v$ for some fixed vector $v$ and all $i \in \mathcal R_v$. Then the algebra $\mathcal A = \mathrm{span}\{ w_i w_j^* : i,j\in \mathcal R_v\}$, which could have non-trivial structure depending on choice of the $w_i$, would be privatized to $P = vv^*$ by $\Phi$. For instance, if $|\mathcal R_v | = k$ and $\{ w_i : i\in \mathcal R_v \}$ is an orthogonal set of (non-zero) vectors, then $\mathcal A$ is unitarily equivalent to $M_k(\mathbb{C})$ and would satisfy $\Phi(A) = \tr(A)P$ for all $A\in \mathcal A$.
\end{example}

We finish by identifying another class of entanglement breaking channels that privatize special types of matrix algebras.

\begin{theorem}\label{samerankthm}
Let $\Phi: M_n(\mathbb{C}) \rightarrow M_n(\mathbb{C})$ be an entanglement breaking channel. Suppose there are mutually orthogonal projections $P_k$, $1 \leq k \leq r$, inside $\mathcal M_{\Phi^\dagger}$ such that the projections $\Phi^\dagger(P_k)$ all have the same rank. Let $\mathcal{A}$ be any $\ast$-subalgebra of $M_r(\C)$ with constant diagonals. Then $\Phi$ privatizes an algebra $\ast$-isomorphic to $\mathcal{A}$.
\end{theorem}

\begin{proof}
Note first that the operators $Q_k = \Phi^\dagger(P_k)$ are indeed projections as each $P_k\in \mathcal M_{\Phi^\dagger}$. Let $s = \mathrm{rank}(Q_k)$ and for each $k$ put $\mathcal W_k = Q_k \mathbb{C}^n$. For each pair $1 \leq k,l \leq r$, let $V_{kl}$ be a partial isometry on $\mathbb{C}^n$ with initial projection $V_{kl}^* V_{kl} = Q_l$ and final projection $V_{kl}V_{kl}^* = Q_k$. The $V_{ij}$ set up a $r\times r$ block matrix picture for operators with domain and range supported on the range of the projection $Q = \sum_{k=1}^r Q_k$.

We can then define a $\ast$-isomorphism $\Psi : M_r(\mathbb{C}) \rightarrow M_n(\mathbb{C})$ by $\Psi(A) = \sum_{k,l = 1}^r a_{kl} V_{kl}$ for matrices $A=(a_{kl})$. Note the image of $\Psi$ inside $M_n(\mathbb{C})$ is unitarily equivalent to $M_r(\mathbb{C}) \otimes I_s$. Moreover, as $\Psi(A) = \sum_{k,l} a_{kl} Q_k V_{kl} Q_l$, by Lemma~\ref{vsorthog} we have
\[
\Phi(\Psi(A)) = \sum_{k=1}^r \Phi(Q_k \Psi(A) Q_k)
= \sum_{k=1}^r a_{kk} \Phi(Q_k)
\]
and $\Phi(Q_k) = P_k \Phi(I) = \Phi(I) P_k$, so that
\[
\Phi(\Psi(A)) = \big( \sum_{k=1}^r a_{kk} P_k \big) \Phi(I) = \Phi(I) \big( \sum_{k=1}^r a_{kk} P_k \big).
\]

Now consider a $\ast$-subalgebra $\mathcal A$ of $M_r(\C)$ with constant diagonals, $a_{kk} = \frac{\tr(A)}{r}$ for all $1\leq k \leq r$ and $A = (a_{kl})\in\mathcal A$. Let $P = \sum_k P_k$. Then from the above calculation we have for all $A\in \mathcal A$,
\[
\Phi(\Psi(A)) = \sum_{k=1}^r \frac{1}{r}\tr(A) P_k \Phi(I)  = \frac{\tr(A)}{r} (P\Phi(I)) = \frac{\tr(A)}{r} (\Phi(I) P).
\]
Hence $\Phi$ privatizes the algebra $\Psi(\mathcal A)$ and the result follows.
\end{proof}

\begin{remark}
Algebras satisfying the condition of having all diagonal entries the same are plentiful and may be generated by taking any partition of the integer $r = \sum_{k=1}^p m_k$, and factoring each part in the partition $m_k = i_kj_k$; then there is a unitary $U$ such that, after conjugation by $U$, the algebra $\oplus_{k=1}^p I_{\max{i_k,j_k}}\otimes M_{\min{i_k,j_k}}(\C)$ will have this form. In particular, for any expression of $r$ as a sum of squares, $r = \sum_{i=1}^p i_k^2$, we find that there is a unitary so that $U\bigl(\oplus_{k=1}^pI_{i_k}\otimes M_{i_k}(\C)\bigr)U^*$ has the property we seek.
\end{remark}

\begin{example}
For an explicit example of a subclass of entanglement breaking channels that satisfy the conditions of the theorem, take a positive integer $n$ with factors $n = rs$.

For each $1 \leq k \leq r$, choose an orthonormal set of vectors $\{ w_{i,k} : 1 \leq i \leq s \} \subseteq \mathbb{C}^n$ and let $Q_k$ be the projection onto the ($s$-dimensional) subspace they span. Then, $I = \sum_k Q_k$ and $Q_k = \sum_{i=1}^s w_i w_i^*$. Next, for each $1\leq k \leq r$ choose sets of non-zero (but not necessarily orthogonal) vectors $\{v_{i,k} : 1 \leq i \leq s\}\subseteq \mathbb{C}^n$, subject to the constraint $v_{i,k}^* v_{j,l} = 0$ whenever $k\neq l$. Let $P_k$ be the projection onto the subspace spanned by $\{v_{i,k}\}_i$.

Let $\Phi$ be the entanglement breaking channel with Kraus operators
\[
\{ v_{i,k}w_{i,k}^* : 1\leq i \leq s, \, 1\leq k \leq r\}.
\]
Then $P_1, \ldots P_r$ are a family of mutually orthogonal projections that belong to the multiplicative domain $\mathcal M_{\Phi}^\dagger$, as they are mapped to projections, in fact $\Phi^\dagger (P_k) = Q_k$. Also recall each $Q_k$ is rank-$s$. Thus, Theorem~\ref{samerankthm} applies, and $\Phi$ privatizes a family of algebras that are isomorphic to subalgebras of $M_r(\mathbb{C})$ with constant diagonals, as per the construction of the proof.
\end{example}

\vspace{0.1in}

{\noindent}{\it Acknowledgements.} D.W.K. was partly supported by NSERC and a University Research Chair at Guelph. R.P. was partly supported by NSERC. M.R is partially supported by Research Initiation Grant at the BITS-Pilani Goa Campus.

\bibliographystyle{plain}
\bibliography{KLOPRbib}

\end{document}